\documentclass[preprint, 3p]{elsarticle}
\usepackage{tikz}
\usepackage[section]{placeins}
\usepackage{bm}
\usepackage{subcaption}
\usepackage{float}
\usepackage{amsmath}
\usepackage{graphicx}
\usepackage{epstopdf}
\usepackage{mathtools}
\usepackage{amssymb}
\usepackage{amssymb,amsmath}
\usepackage{amsgen,amsfonts,amsbsy,amsthm}
\usepackage{latexsym}
\usepackage{times}
\usepackage{algorithm}
\usepackage[noend]{algpseudocode}
\allowdisplaybreaks
\makeatletter
\def\BState{\State\hskip-\ALG@thistlm}
\makeatother

\newcommand{\bvec}[1]{\bm{#1}}
\newcommand{\real}{\mathbb{R}}

\newcommand{\supp}[1]{\texttt{supp}\left(#1\right)}
\newcommand{\norm}[1]{\left\|#1\right\|_2}
\newcommand{\opnorm}[2]{\left\|#1\right\|_{#2}}
\newcommand{\abs}[1]{\left|#1\right|}
\newcommand{\diag}[1]{\texttt{diag}\left(#1\right)}
\newcommand{\prob}[1]{\mathbb{P}\left(#1\right)}
\newcommand{\expect}[1]{\mathbb{E}\left[#1\right]}
\newcommand{\expectsuff}[2]{\mathbb{E}_{#1}\left[#2\right]}

\newtheorem{thm}{Theorem}[section]

\newtheorem{rmk}{Remark}[section]

\journal{Elsevier Signal Processing}
\begin{document}
\setlength\abovedisplayskip{0pt}
\setlength\belowdisplayskip{0pt}
\begin{frontmatter}
\title{On the MMSE Estimation of Norm of a Gaussian Vector under Additive White Gaussian Noise with Randomly Missing Input Entries}
\setlength{\abovedisplayskip}{0pt}
\setlength{\belowdisplayskip}{0pt}
\vspace{-5mm}
\author{Samrat Mukhopadhyay\corref{cor}}
\address{Department of Electronics and Electrical Communication
	Engineering, Indian Institute of Technology, Kharagpur, INDIA }
\ead{samratphysics@gmail.com}
\cortext[cor]{Corresponding Author}



    %
\begin{abstract}
This paper considers the task of estimating the $l_2$ norm of a $n$-dimensional random Gaussian vector from noisy measurements taken after many of the entries of the vector are \emph{missed} and only $K\ (0\le K\le n)$ entries are retained and others are set to $0$. Specifically, we evaluate the minimum mean square error (MMSE) estimator of the $l_2$ norm of the unknown Gaussian vector performing  measurements under additive white Gaussian noise (AWGN) on the vector after the data missing and derive expressions for the corresponding mean square error (MSE). We find that the corresponding MSE normalized by $n$ tends to $0$ as $n\to \infty$ when $K/n$ is kept constant. Furthermore, expressions for the MSE is derived when the variance of the AWGN noise tends to either $0$ or $\infty$. These results generalize the results of Dytso et al.~\cite{dytso2019estimating} where the case $K=n$ is considered, i.e. the MMSE estimator of norm of random Gaussian vector is derived from measurements under AWGN noise without considering the data missing phenomenon. 
\end{abstract}
\end{frontmatter}
\section{Introduction}
\label{sec:intro}

\begin{figure}[ht!]
	\centering
	\includegraphics[height=0.5in,width=2in]{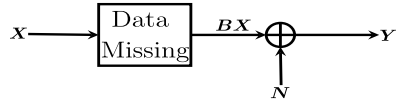}
	\caption{Block diagram of the measurement model.}
	\label{fig:measurement-model}
\end{figure}
In this paper we consider the estimation of the $l_2$ norm of a random vector with missing entries under measurement noise. In particular, the measurement model that we consider in this paper is the following: \begin{align}
\label{eq:measurement-model}
\bvec{Y} & =\bvec{B X} + \bvec{N},
\end{align}
where $\bvec{X}\sim\mathcal{N}(\bvec{0},\bvec{I}_n),\ \bvec{N}\sim\mathcal{N}(\bvec{0},\sigma^2\bvec{I}_n)$, with $\bvec{X}$ being the input vector and $\bvec{N}$ the measurement noise vector, which are independent. Furthermore, $\bvec{B}=\diag{B_1,\cdots,B_n},$ where $B_i\in\{0,1\}$ and $\sum_{i=1}^n B_i=K$, for a given non-negative integer $K\le n$. The data missing process is assumed to be independent of $\bvec{X}$ and $\bvec{N}$, i.e. $\bvec{B}$ is independent of $\bvec{X},\bvec{N}$. The block diagram of the measurement process is shown in Fig.~\ref{fig:measurement-model}. 

The problem of estimating $l_2$ norm of a vector from noisy measurements can have different applications in domains such as \emph{distributed computing}~\cite{dwork2014algorithmic}, \emph{wireless systems}~\cite{zhang2007user} and \emph{wireless networks}~\cite{boyd2006randomized}, where different nodes in the system might want to estimate the $l_2$ norm of some vector of interest, either in a centralized or distributed manner in order to perform different tasks such as secure transmission, selection of best transmitting node etc. Dytso et al.~\cite{dytso2019estimating} have studied the optimal MMSE estimator for the $l_2$ norm of a random Gaussian vector in additive white Gaussian noise when no entry of the input vector is missing. However, the problem of \emph{missing input data} entries at random instances is quite common in many signal processing and communication applications~\cite{chen2000system,fang2009parameter,loh2011high,chen2013robust,mukhopadhyay20-semea,mukhopadhyay20-imdlms}. Such missing input data model is also prevalent in information and coding theory where \emph{deletion channels}~\cite{mitzenmacher2009survey} are used to model a communication channel where sequence of input bits are deleted at random locations.  Furthermore, the measurement model of Eq.~\eqref{eq:measurement-model} is also useful in the context of applications where storage is scarce so that only a few of the input vectors can be loaded while the others have to be considered $0$.  

In this work we find explicit expressions for the minimum mean square error (MMSE) estimator of $\norm{\bvec{X}}$ given the measurements according to the model~\eqref{eq:measurement-model}. In particular, we obtain an expression for the conditional expectation of $\norm{\bvec{X}}$ given the measurement vector $\bvec{Y}$. We further characterize the mean square error (MSE) for this optimal estimator for $0\le K\le n$ and demonstrate using numerical simulations that even when $K$ is small, the MSE is close to the MSE corresponding to the case $K=n$. Our main contribution can be summarized as the generalization of the results of~\cite{dytso2019estimating} by characterizing the MMSE estimator in the presence of the data missing phenomenon where only $0\le K\le N$ entries are retained and showing that even with small $K$, the MMSE estimation of a Gaussian vector under AWGN noise is pretty accurate, especially for large input size $n$. 

In the following, let $[n]:=\{1,\cdots,n\}$. For any vector $\bvec{x}\in \real^n$ and subset $S\subset [n]$, $\bvec{x}_S$ denotes the vector $\bvec{x}$ restricted to $S$, i.e., $\bvec{x}_S$ contains those entries of $\bvec{x}$ that are indexed by $S$. Let $\supp{\bvec{x}}$ denote the support of $\bvec{x}$, i.e., $[\bvec{x}]_i=0$ if $i\notin \supp{\bvec{x}}$. We denote by $\bvec{I}_n$, the identity matrix of size $n\times n$ and $\bvec{I}_S$ denotes the submatrix containing the columns of $\bvec{I}_n$ indexed by $S$. We denote by $\prob{\cdot},\ \expect{\cdot}$ the probability and expectation operators, respectively, and by $\expectsuff{R}{\cdot}$ we denote expectation with respect to a random variable $R$. Let $\delta(\cdot)$ denote the Dirac delta function, $_pF_q$ denote the generalized hypergeometric function~\cite{olver2010nist}, defined as $_pF_q(a_1,\cdots,a_p;b_1,\cdots,b_q;x)=\sum_{k=0}^\infty \frac{(a_1)_k\cdots(a_p)_k}{(b_1)_k\cdots(b_q)_k}\frac{x^k}{k!}$, where the Pochhammer symbol~\cite{olver2010nist} $(a)_k$ is defined as $(a)_k=a(a+1)\cdots(a+k-1)$ for any integer $k\ge 1$, and $(a)_0=1$. We denote by $f_X(\cdot)$ the probability density function of any continuous random variable $X$. $B(\cdot,\cdot)$ denotes the beta function defined as $B(m,n)=\frac{\Gamma(m)\Gamma(n)}{\Gamma(m+n)}$, where $\Gamma(\cdot)$ is the gamma function. 
\section{MMSE Estimator and Its Characterization}
\label{sec:mmse-estimator-characterization}
In this section we present our results on the MMSE estimator of $\norm{\bvec{X}}$ given the measurement $\bvec{Y}$ according to the measurement model~\eqref{eq:measurement-model}. The following Theorem presents a closed form expression of the MMSE estimator.
\begin{thm}
	\label{thm:mmse-estimator}
	Let $\bvec{X},\bvec{Y},\bvec{B},\bvec{N}$ be as described in the measurement model~\eqref{eq:measurement-model}. Let $\prob{\supp{\bvec{B}}=S} = 1/\binom{n}{K}$ for any $S\subset [n]$ with $\abs{S}=K$ and $J=n-K$. Then, \begin{align}
	\label{eq:mmse-estimator-expectation-expression-given-B}
	\expect{\norm{\bvec{X}}\mid \bvec{Y}=\bvec{y},\bvec{B}=\bvec{b}} & = \expectsuff{U_1,U_2}{\sqrt{\frac{\sigma^2}{\sigma^2+1}U_1+U_2}}\\
	\label{eq:mmse-estimator-series-expression-given-B}
	\ & =
	 \sqrt{2}\left(\frac{\sigma^2}{\sigma^2+1}\right)^{\frac{J+1}{2}}e^{-\frac{\norm{\bvec{y}_S}^2}{2\sigma^2(\sigma^2+1)}}\sum_{l=0}^\infty a_l\left(\frac{\norm{\bvec{y}_S}^{2}}{2\sigma^2(\sigma^2+1)}\right)^l,
	\end{align} where, with $\supp{\bvec{b}}=S$, $U_1,U_2$ are independent; $U_1$ has a noncentral $\chi^2$ distribution with noncentrality parameter $\lambda = \frac{\norm{\bvec{y}_S}^2}{\sigma^2(\sigma^2+1)}$ and degrees of freedom $K$, while $U_2$ has a central $\chi^2$ distribution with degrees of freedom $J$.
	Consequently, \begin{align}
	\label{eq:mmse-estimator-expression}
	\expect{\norm{\bvec{X}}\mid \bvec{Y}=\bvec{y}} & = \frac{\sqrt{2}\left(\frac{\sigma^2}{\sigma^2+1}\right)^{\frac{J+1}{2}}}{\binom{n}{K}}\sum_{l=0}^\infty a_l\sum_{\substack{S\subset [n]:\\ \abs{S}=K}}e^{-\frac{\norm{\bvec{y}_S}^2}{2\sigma^2(\sigma^2+1)}} \left(\frac{\norm{\bvec{y}_S}^{2}}{2\sigma^2(\sigma^2+1)}\right)^l\\
	\label{eq:mmse-estimator-alternate-expression-as-expectation}
	\ &
	 =\sqrt{2}\left(\frac{\sigma^2}{\sigma^2+1}\right)^{\frac{J+1}{2}}\sum_{l=0}^\infty a_l\expectsuff{S}{e^{-Z_S} Z_S^l\mid \bvec{Y}=\bvec{y}},
	\end{align}
	where $Z_S=\frac{\norm{\bvec{Y}_S}^2}{2\sigma^2(\sigma^2+1)}$, and \begin{align}
	\label{eq:al-definition}
	a_l & = \frac{1}{l!}\sum_{k=0}^\infty \frac{\left(\frac{J}{2}\right)_k\Gamma\left(\frac{n+1}{2}+k+l\right)}{\Gamma\left(\frac{n}{2}+k+l\right)}\frac{1}{k!(\sigma^2+1)^k}.
	\end{align}
\end{thm}
\begin{proof}
	The proof is deferred to Section~\ref{sec:proof-thm-mmse-estimator}.
\end{proof}
\begin{rmk}
	The Theorem~\ref{thm:mmse-estimator} provides several alternate descriptions of the MMSE estimator. First of all, while Eq.~\eqref{eq:mmse-estimator-expectation-expression-given-B} enables one to evaluate the conditional expectation $\expect{\norm{\bvec{X}}\mid \bvec{Y}=\bvec{y},\bvec{B}=\bvec{b}}$ by Monte-Carlo simulations, Eq.~\eqref{eq:mmse-estimator-series-expression-given-B} can be used to find the estimator by evaluating terms of a series expansion. Furthermore, while Eq.~\eqref{eq:mmse-estimator-expression} is a simple consequence of Eq.~\eqref{eq:mmse-estimator-series-expression-given-B}, Eq.~\eqref{eq:mmse-estimator-alternate-expression-as-expectation} gives another interpretation of the MMSE estimator as an expectation over randomly chosen subsets $S\subset[n]$ such that $\abs{S}=K$.  
\end{rmk}
\begin{figure}[t!]
		\centering
		\includegraphics[height=2in,width=3.5in]{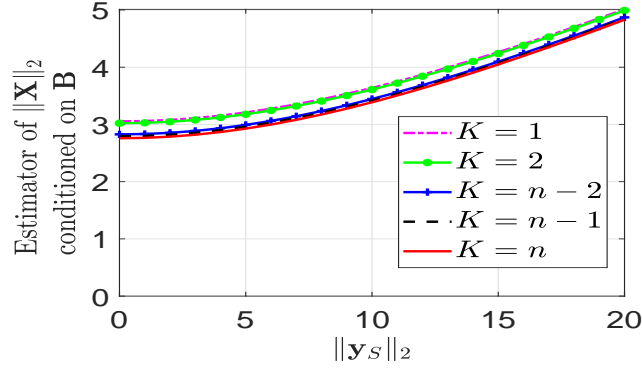}
		\caption{Comparison of MMSE estimator conditioned on $\bvec{B}$ for different values of $K$ where $\abs{S}=K$. Here $n=10,\sigma=2$.}
		\label{fig:mse-conditional-estimator-vs-norm-y}
\end{figure}
	%
%
The conditional MMSE estimator $\expect{\norm{\bvec{X}}\mid \bvec{Y}=\bvec{y},\bvec{B} = \bvec{b}}$ is evaluated for different values of $\norm{\bvec{y}_S}$, where $\abs{S}=K$ and the results are plotted in Fig.~\ref{fig:mse-conditional-estimator-vs-norm-y} where it can be observed that even with small number of actual measurements $K (K=1,2)$, the MMSE estimator is pretty close to the one obtained when actual measurement of all the $n$ coordinates are considered. 

We now proceed to calculate the MSE corresponding to the MMSE estimator derived in Eq.~\eqref{eq:mmse-estimator-expression}. The following Theorem presents an expression for the MSE for the estimator in Eq.~\eqref{eq:mmse-estimator-expression}, which is now on referred to as \texttt{mmse}.
\begin{thm}
	\label{thm:mmse-expression}
	Let $\bvec{X},\bvec{Y},\bvec{B},\bvec{N}$ be described as in model~\eqref{eq:measurement-model} and $\prob{\supp{\bvec{B}}=S} = 1/\binom{n}{K}$ for any $S\subset [n]$ with $\abs{S}=K$ and $J=n-K$. Then,
	\begin{align}
	\label{eq:mmse-expression}
	\texttt{mmse} & = n - \frac{2\left(\frac{\sigma^2}{\sigma^2+1}\right)^{J+1}}{\binom{n}{K}}\sum_{\substack{r_{\min}\le r\le K,,\\i,j\ge 0}}\binom{K}{r}\binom{n-K}{K-r}a_ia_jh_{ij}(r),
	\end{align}
	where $r_{\min}=\max\{0,2K-n\}$ and
	\begin{align}
	\label{eq:h_ij-definition}
	h_{ij}(r) & = \left\{
	\begin{array}{ll}
		\frac{\sigma^{2K}\Gamma\left(i+\frac{K}{2}\right)\Gamma\left(j+\frac{K}{2}\right)}{(\sigma^2+1)^{K+i+j}\left(\Gamma\left(\frac{K}{2}\right)\right)^2}, & r = 0,\\
		\sum_{\substack{0\le s\le i,\\ 0\le t\le j}} \binom{i}{s}\binom{j}{t} \frac{\sigma^{2K-r}\Gamma\left(s+\frac{K-r}{2}\right)\Gamma\left(t+\frac{K-r}{2}\right)\Gamma\left(i+j-(s+t)+\frac{r}{2}\right)}{(\sigma^2+1)^{s+t+K-r}(\sigma^2+2)^{i+j-(s+t)+\frac{r}{2}}\left(\Gamma\left(\frac{K-r}{2}\right)\right)^2\Gamma\left(\frac{r}{2}\right)}, & 1\le r\le K-1,\\
		\frac{\sigma^K\Gamma\left(i+j+\frac{K}{2}\right)}{(\sigma^2+2)^{i+j+K/2}\Gamma\left(\frac{K}{2}\right)}, & r = K.
	\end{array}
	\right.
	\end{align}
\end{thm}
\begin{proof}
	The proof is deferred to Section~\ref{sec:proof-thm-mmse-expression}.
\end{proof}
\begin{rmk}
	If $K=n$, $J=0$ and $r_{\min}=n$, from Eq.~\eqref{eq:mmse-expression}, $\texttt{mmse}=n-\frac{2\sigma^2}{(\sigma^2+1)}\sum_{i,j\ge 0}a_ia_jh_{ij}(n)$. Using the expression of $a_i$ from Eq.~\eqref{eq:al-definition} and $h_{ij}(K)$ from Eq.~\eqref{eq:h_ij-definition}, it can be verified that we obtain the same expression as that of Theorem 2 of~\cite{dytso2019estimating}. This establishes the result of Theorem~\ref{thm:mmse-expression} as a generalization of Theorem 2 of~\cite{dytso2019estimating}.
\end{rmk}
\begin{figure}[t!]
	\centering
	\includegraphics[height=2in,width=3.5in]{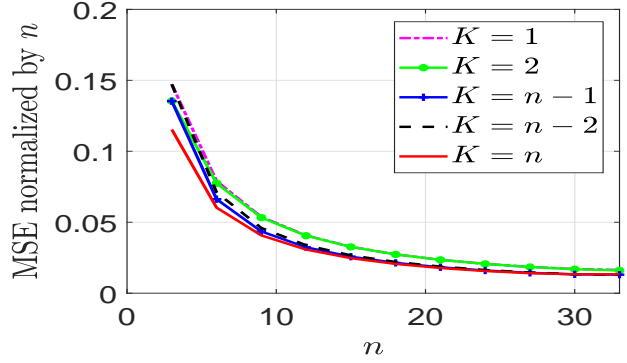}
	\caption{Comparison of MSE normalized by $n$ for different values of $K$ with $\sigma=1$.}
	\label{fig:mmse-vs-n}
	
\end{figure}
The MSE for the estimators for different values of $K,n$ are plotted in Fig.~\ref{fig:mmse-vs-n} where it is observed that even with small $K$ the \texttt{mmse} is pretty close to the \texttt{mmse} with $K=n$ which again.

We conclude this section with some results on the asymptotic behavior of the MMSE: 
\begin{thm}
	\label{thm:asymptotic-mmse}
	Let $\bvec{X},\bvec{Y},\bvec{B},\bvec{N}$ be described as in model~\eqref{eq:measurement-model} and $\prob{\supp{\bvec{B}}=S} = 1/\binom{n}{K}$ for any $S\subset [n]$ with $\abs{S}=K$ and $J=n-K$.  Then, \begin{enumerate}
		\item For fixed $K,n$, $\lim_{\sigma\to 0}\texttt{mmse} = n - (n+J)\frac{B\left(\frac{n+J}{2},\frac{n+1}{2}\right)}{B\left(\frac{n+J+1}{2},\frac{n}{2}\right)}\ _3F_2\left(\frac{n+1}{2},\frac{J}{2},-\frac{1}{2};\frac{n+J+1}{2},\frac{n}{2};1\right)$.
		\item For fixed $K,n$, $\lim_{\sigma\to \infty}\texttt{mmse} = n - 2\left(\frac{\Gamma\left(\frac{n+1}{2}\right)}{\Gamma\left(\frac{n}{2}\right)}\right)^2$.
		\item For fixed $\sigma>0$ and $K,n\to \infty$ with $K/n=p\in [1/n,1]$, $\lim_{n\to \infty}\frac{\texttt{mmse}}{n}=0$.
	\end{enumerate}
\end{thm}
\begin{proof}
	The proof is deferred to Section~\ref{sec:proof-thm-asymptotic-mmse}.
\end{proof}
\begin{rmk}
	The first result shows that when $\sigma\to 0$, there is a nonzero \texttt{mmse}  arising entirely due to the randomness associated to partial entries of $\bvec{X}$ due to data missing. Note that if $J=0$ so that no entries of $\bvec{X}$ is missing, the expression for $\texttt{mmse}$ reduces to $0$ when $\sigma\to 0$, which is intuitively satisfying. The second result shows that when $\sigma\to \infty$, the \texttt{mmse} is identical to the second result of Theorem 3 of~\cite{dytso2019estimating}. The third result of Theorem~\ref{thm:asymptotic-mmse} shows that as long as $K$ is such that although it is large but the ratio of $K/n$ is between $1/n$ to $1$, the ratio $\texttt{mmse}/n$ tends to $0$, which is interesting since it shows that even with small ratio $K/n$, the estimator yields near similar performance as the one which uses measurements of all coordinates. 
\end{rmk}
%
%
%
%
\section{Proof of Theorems}
\label{sec:proofs}
\subsection{Proof of Theorem~\ref{thm:mmse-estimator}}
\label{sec:proof-thm-mmse-estimator}
%
 Note that $\expect{\norm{\bvec{X}}\mid \bvec{Y}=\bvec{y},\ \bvec{B}=\bvec{b}} = \frac{\int_0^\infty tf_{\norm{\bvec{X}},\bvec{Y}\mid \bvec{B}}\left(t,\bvec{y}\mid \bvec{b}\right)dt}{f_{\bvec{Y}\mid \bvec{B}}(\bvec{y}\mid \bvec{b})}$. Now,
\begin{align}
f_{\norm{\bvec{X}},\bvec{Y}\mid \bvec{B}}\left(t,\bvec{y}\mid \bvec{b}\right)
\ & = \int_{\real^N} f_{\norm{\bvec{X}},\bvec{Y}\mid \bvec{B},\bvec{X}}\left(t,\bvec{y}\mid \bvec{b},\bvec{x}\right)f_{\bvec{X}}(\bvec{x}) d\bvec{x} = \int_{\real^N} f_{\bvec{Y}\mid \bvec{B},\bvec{X}}\left(\bvec{y}\mid \bvec{b},\bvec{x}\right)\delta\left(\norm{\bvec{x}}-t\right)f_{\bvec{X}}(\bvec{x}) d\bvec{x},
\end{align} 
which follows from the observation that given $\bvec{X}$, $\norm{\bvec{X}}$ is fixed, so that it is independent of $\bvec{Y}$. Now, given $\bvec{B}$ such that $\supp{\bvec{B}}=S$ with $\abs{S}=K,\ J=n-K$, we get $\bvec{Y}_{S^C}\mid \bvec{X}=\bvec{N}_{S^C}$ and $\bvec{Y}_S\mid \bvec{X}=\bvec{X}_S+\bvec{N}_S$, so that $\bvec{Y}_S\mid \bvec{X}\sim \mathcal{N}(\bvec{X}_S,\sigma^2 \bvec{I}_S)$, $\bvec{Y}_{S^C}\mid \bvec{X}\sim \mathcal{N}(\bvec{0}_{S^C},\sigma^2\bvec{I}_{S^C})$ and $\bvec{Y}_S\sim \mathcal{N}(\bvec{0}_S,(\sigma^2+1) \bvec{I}_S),\ \bvec{Y}_{S^C}\sim \mathcal{N}(\bvec{0}_{S^C},\sigma^2\bvec{I}_{S^C})$. Consequently, \begin{align}
f_{\norm{\bvec{X}},\bvec{Y}\mid \bvec{B}}\left(t,\bvec{y}\mid \bvec{b}\right)
\ & =\frac{e^{-\frac{\norm{\bvec{y}_{S^C}}^2}{2\sigma^2}}}{(2\pi\sigma)^{n}}\int_{\real^n}e^{-\left(\frac{\norm{\bvec{y}_S-\bvec{x}_S}^2}{2\sigma^2}+\frac{\norm{\bvec{x}}^2}{2}\right)}\delta\left(\norm{\bvec{x}}-t\right) d\bvec{x}\nonumber\\
\ & = \frac{e^{-\left(\frac{\norm{\bvec{y}_{S^C}}^2}{2\sigma^2}+\frac{\norm{\bvec{y}_{S}}^2}{2(\sigma^2+1)}\right)}}{(2\pi(\sigma^2+1))^{K/2}((2\pi\sigma^2)^{J/2}} \int_{\real^N}\frac{e^{-\left(\frac{\norm{\bvec{x}_S-\frac{\bvec{y}_S}{\sigma^2+1}}^2}{\frac{2\sigma^2}{\sigma^2+1}}+\frac{\norm{\bvec{x}_{S^C}}^2}{2}\right)}}{\left(2\pi\left(\frac{\sigma^2}{\sigma^2+1}\right)\right)^{K/2}(2\pi)^{J/2}}\delta(\norm{\bvec{x}}-t)d\bvec{x}\nonumber\\
\ & =f_{\bvec{Y}\mid \bvec{B}}(\bvec{y}\mid \bvec{b})f_{Z}(t),
\end{align}   
where the penultimate step above uses the following identity: $\displaystyle \frac{\norm{\bvec{y}_S-\bvec{x}_S}^2}{2\sigma^2}+\frac{\norm{\bvec{x}}^2}{2} =\frac{\norm{\bvec{x}_{S^C}}^2}{2}+ \frac{\sigma^2+1}{2\sigma^2}\norm{\bvec{x}_S-\frac{\bvec{y}_S}{\sigma^2+1}}^2+\frac{\norm{\bvec{y}_S}^2}{2(\sigma^2+1)},$
and where we define the random variable $Z=\norm{\widetilde{\bvec{X}}}$, $\widetilde{\bvec{X}}\sim \mathcal{N}\left(\bvec{\mu},\bvec{\Sigma}\right)$, with \begin{align}
\bvec{\mu} = \begin{bmatrix}
\frac{\bvec{y}_S}{\sigma^2+1} \\
\bvec{0}_{S^C}
\end{bmatrix},\ 
\bvec{\Sigma} & = \begin{bmatrix}
\frac{\sigma^2}{\sigma^2+1}\bvec{I}_S & \bvec{O}_{S^C}\\
\bvec{O}_S & \bvec{I}_{S^C}
\end{bmatrix}.
\end{align}
Therefore, 
\begin{align}
\label{eq:mmse-expression-conditional-simple-expectation}
\expect{\norm{\bvec{X}}\mid \bvec{Y}=\bvec{y},\ \bvec{B}=\bvec{b}} & = \int_0^\infty t f_{Z}(t)dt=\expect{Z}.
\end{align}
Note that one can write $Z^2=\frac{\sigma^2}{\sigma^2+1}\norm{\bvec{X}_1}^2+\norm{\bvec{X}_2}^2$, where $\bvec{X}_1\sim\mathcal{N}\left(\frac{\bvec{y}_S}{\sqrt{\sigma^2(\sigma^2+1)}},\bvec{I}_S\right)$ and $\bvec{X}_2\sim\mathcal{N}(\bvec{0}_{S^C},\bvec{I}_{S^C})$. Therefore, $Z^2=\frac{\sigma^2}{\sigma^2+1}U_1+U_2$ with the descriptions of $U_1,U_2$ provided in Theorem~\ref{thm:mmse-estimator} and consequently Eq.~\eqref{eq:mmse-estimator-expectation-expression-given-B} follows. 

To further simplify Eq.~\eqref{eq:mmse-estimator-expectation-expression-given-B}, writing $V=Z^2$, one obtains, for all $v\ge 0$, \begin{align}
\label{eq:convolution-expression-prelim}
f_{V}(v) & = f_{\frac{\sigma^2}{\sigma^2+1}U_1}(v)\circledast f_{U_2}(v)
=\frac{\sigma^2+1}{\sigma^2}\int_0^v f_{U_1}\left(\frac{(\sigma^2+1)u}{\sigma^2}\right)f_{U_2}(v-u)du.
\end{align}
Now we recall that for all $u\ge 0$, \begin{align}
\label{eq:noncentral-chi-square-pdf}
f_{U_1}(\lambda,K;u) & = \frac{e^{-\lambda/2}{}_0F_1(;K/2;\lambda u/4)e^{-u/2}u^{K/2-1}}{2^{K/2}\Gamma(K/2)},\\
\label{eq:central-chi-square-pdf}
f_{U_2}(J;u) & = \frac{e^{-u/2}u^{J/2-1}}{\Gamma(J/2)2^{J/2}}.
\end{align}
Therefore, from~\eqref{eq:convolution-expression-prelim}, using the expressions from Eqs.~\eqref{eq:noncentral-chi-square-pdf} and~\eqref{eq:central-chi-square-pdf}, and using a substitution of variable inside the integration as $u\to \rho$ with $u=\rho v$, as well as using the identity $\lambda(\sigma^2 + 1)/(4\sigma^2)=\norm{\bvec{y}_S}^2/(4\sigma^4)$, it follows after some straightforward simplifications, that \begin{align}
f_V(v)  & =\frac{\left(\frac{\sigma^2+1}{\sigma^2}\right)^{K/2}e^{-(\lambda+v)/2}v^{n/2-1}}{2^{n/2}\Gamma(J/2)\Gamma(K/2)} \int_0^1 e^{-\frac{\rho v}{2\sigma^2}}{\rho}^{K/2-1}(1-\rho)^{J/2-1}{}_0F_1\left(;K/2;\frac{\norm{\bvec{y}_S}^2v\rho}{4\sigma^4}\right)d\rho\nonumber\\
\ & = \frac{\left(\frac{\sigma^2+1}{\sigma^2}\right)^{K/2}e^{-(\lambda+v)/2}v^{n/2-1}}{2^{n/2}\Gamma(J/2)\Gamma(K/2)} R\left(\frac{v}{2\sigma^2},K/2,J/2,K/2,\frac{\norm{\bvec{y}_S}^2v\rho}{4\sigma^4}\right),
\end{align}
where we define, for $\alpha,\beta,\gamma,\nu,\epsilon\ge 0$, \begin{align}
\label{eq:r-function}
R(\alpha,\beta,\gamma,\nu,\epsilon) & = \int_0^1 e^{-\alpha x}x^{\beta-1}(1-x)^{\gamma-1}{}_0F_1(;\nu;\epsilon x)dx.
\end{align}
In~\ref{sec:appendix-evaluation-of-r-function} the following is established: \begin{align}
\label{eq:r-function-evaluation}
R(\alpha,\beta,\gamma,\beta,\epsilon) & = e^{-\alpha}B(\beta,\gamma)\sum_{k,l\ge 0}\frac{\alpha^k\epsilon^l(\gamma)_k}{k!l!(\beta+\gamma)_{k+l}},
\end{align}
Therefore, one obtains, \begin{align}
f_V(v) &  =\frac{\left(\frac{\sigma^2+1}{\sigma^2}\right)^{K/2}e^{-(\lambda+v)/2}v^{n/2-1}}{2^{n/2}\Gamma(J/2)\Gamma(K/2)} e^{-\frac{v}{2\sigma^2}}B(K/2,J/2)\sum_{k,l\ge 0}\frac{\left(\frac{v}{2\sigma^2}\right)^k\left(\frac{\norm{\bvec{y}_S}^2v}{4\sigma^4}\right)^l(J/2)_k}{(n/2)_{k+l}}\nonumber\\
\ & = \frac{\left(\frac{\sigma^2+1}{\sigma^2}\right)^{K/2}e^{-(\lambda+v(1+1/\sigma^2))/2}v^{n/2-1}}{2^{n/2}\Gamma(n/2)}\sum_{k,l\ge 0}\frac{v^{k+l}\norm{\bvec{y}_S}^{2l}}{k!l!(2\sigma^2)^{k+2l}}\frac{(J/2)_k}{(n/2)_{k+l}}.
\end{align}

Consequently, \begin{align}
\lefteqn{\expect{\norm{\bvec{X}}\mid \bvec{Y}=\bvec{y},\ \bvec{B}=\bvec{b}} = \int_0^\infty t\cdot 2tf_V(t^2) dt} & &\nonumber\\
\ & = \int_0^\infty 2t^2 \frac{\left(\frac{\sigma^2+1}{\sigma^2}\right)^{K/2}e^{-(\lambda+t^2(1+1/\sigma^2))/2}t^{n-2}}{2^{n/2}\Gamma(n/2)} \sum_{k,l\ge 0}\frac{t^{2(k+l)}\norm{\bvec{y}_S}^{2l}}{k!l!(2\sigma^2)^{k+2l}}\frac{(J/2)_k}{(n/2)_{k+l}} dt\nonumber\\
\ & = \sum_{k,l\ge 0}\frac{\left(\frac{\sigma^2+1}{\sigma^2}\right)^{K/2}e^{-\lambda/2}}{2^{n/2-1}\Gamma(n/2)}\cdot \frac{\norm{\bvec{y}_S}^{2l}}{k!l!(2\sigma^2)^{k+2l}}\cdot \frac{(J/2)_k}{(n/2)_{k+l}}\int_0^\infty t^{n+2(k+l)}e^{-t^2(1+1/\sigma^2)/2}dt.
\end{align}
For any $a,b>0$, it is straightforward to deduce that,$\int_0^\infty t^ae^{-bt^2}dt = \frac{\Gamma(\frac{a+1}{2})}{2b^{(a+1)/2}}$, so that, \begin{align}
\expect{\norm{\bvec{X}}\mid \bvec{Y}=\bvec{y},\ \bvec{B}=\bvec{b}} & =\frac{\left(\frac{\sigma^2+1}{\sigma^2}\right)^{K/2}e^{-\lambda/2}}{2^{n/2-1}\Gamma(n/2)}\sum_{k,l\ge 0}\frac{\norm{\bvec{y}_S}^{2l}}{k!l!(2\sigma^2)^{k+2l}} \frac{(J/2)_k}{(n/2)_{k+l}} \frac{1}{2}\left(\frac{2\sigma^2}{\sigma^2+1}\right)^{\frac{n+1}{2}+k+l}\Gamma\left(\frac{n+1}{2}+k+l\right)\nonumber\\
\ & = \sqrt{2}\left(\frac{\sigma^2}{\sigma^2+1}\right)^{\frac{J+1}{2}}e^{-\lambda/2} \sum_{k,l\ge 0}\frac{\norm{\bvec{y}_S}^{2l}}{(2\sigma^2)^l(\sigma^2+1)^{k+l} k!l!} \frac{\left(\frac{J}{2}\right)_k\Gamma(\frac{n+1}{2}+k+l)}{\Gamma\left(\frac{n}{2}+k+l\right)},
\end{align}
which results in Eq.~\eqref{eq:mmse-estimator-series-expression-given-B} after plugging in the expressions of $a_l,\ l\ge 0$ from Eq.~\eqref{eq:al-definition}. Thereafter, taking expectation over the realizations of $\bvec{B}$, the Eqs.~\eqref{eq:mmse-estimator-expression} and~\eqref{eq:mmse-estimator-alternate-expression-as-expectation} follow. 
\subsection{Proof of Theorem~\ref{thm:mmse-expression}}
\label{sec:proof-thm-mmse-expression}
The expression for MSE of the estimator $\expect{\norm{\bvec{X}}\mid \bvec{Y}=\bvec{y}}$ is \begin{align}
\label{eq:mmse-basic-expression}
\texttt{mmse} & = \expect{\norm{\bvec{X}}^2}-\expect{\left(\expect{\norm{\bvec{X}}\mid \bvec{Y}=\bvec{y}}\right)^2} = n - \expect{\left(\expect{\norm{\bvec{X}}\mid \bvec{Y}=\bvec{y}}\right)^2}.
\end{align}
Using Eq.~\eqref{eq:mmse-estimator-alternate-expression-as-expectation}, it follows that \begin{align}
\expect{\left(\expect{\norm{\bvec{X}}\mid \bvec{Y}=\bvec{y}}\right)^2}
\ & = 2\left(\frac{\sigma^2}{\sigma^2+1}\right)^{J+1}\sum_{i,j=0}^{\infty}a_ia_j\expectsuff{\bvec{Y}}{\expectsuff{S}{Z_S^ie^{-Z_S}\mid\bvec{Y}}\expectsuff{T}{Z_T^je^{-Z_T}\mid \bvec{Y}}}\nonumber\\
\label{eq:mmse-second-term-preliminary-expression}
\ & = \frac{2\left(\frac{\sigma^2}{\sigma^2+1}\right)^{J+1}}{\binom{n}{K}^2}\sum_{i,j=0}^{\infty}a_ia_j\sum_{\substack{S,T\subset [n]:\\\abs{S}=\abs{T}=K}}\expectsuff{\bvec{Y}}{Z_S^iZ_T^je^{-(Z_S+Z_T)}}.
\end{align}
For any two subsets $S,T\subset [n]$ such that $\abs{S}=\abs{T}=K$, let $P=S\setminus T,\ Q=T\setminus S,\ R=S\cap T$. If $\abs{R}=r$,$\ \abs{P}=\abs{Q}=K-r$, and $Z_S=Z_P+Z_R,\ Z_T=Z_Q+Z_R$. Since $Z_S=\frac{\norm{\bvec{Y}_S}^2}{2\sigma^2(\sigma^2+1)}$ and $\bvec{Y}_S\sim \mathcal{N}\left(\bvec{0}_S,(\sigma^2+1)\bvec{I}_S\right)$, it follows that if $1\le r\le K-1$, $Z_P,\ Z_Q,\ Z_R$ are all independent and $2\sigma^2Z_P\sim \chi^2_{K-r},2\sigma^2Z_Q\sim \chi^2_{K-r},2\sigma^2Z_R\sim \chi^2_{r}$. Similarly, if $r=K$, $Z_P=Z_Q=0$ and $2\sigma^2Z_R\sim \chi^2_{K}$, and if $r=0$, $Z_R=0$, and $Z_P,\ Z_Q$ are independent and identically distributed with $2\sigma^2Z_P\sim \chi^2_{K}$.
%
Let $h_{ij}(r)$ denote $\expect{Z_S^iZ_T^je^{-(Z_S+Z_T)}}$ when $\abs{S\cap T}=r,\ 0\le r\le K$.
Then we find, \begin{align}
\ & h_{ij}(0) = \expect{Z_P^{i}e^{-Z_P}}\expect{Z_Q^{j}e^{-Z_Q}}
= \frac{\sigma^{2K}\Gamma\left(i+\frac{K}{2}\right)\Gamma\left(j+\frac{K}{2}\right)}{(\sigma^2+1)^{K+i+j}\left(\Gamma\left(\frac{K}{2}\right)\right)^2},\\
\ & h_{ij}(K) =\expect{Z_R^{i+j}e^{-2Z_R}} = \frac{\sigma^K\Gamma\left(i+j+\frac{K}{2}\right)}{(\sigma^2+2)^{i+j+K/2}\Gamma\left(\frac{K}{2}\right)},\\
\ & h_{ij}(r) = \expect{(Z_P+Z_R)^i(Z_Q+Z_R)^je^{-(Z_P+Z_Q+2Z_R)}}\ (1\le r\le K-1)\nonumber\\
\ & = \sum_{\substack{0\le s\le i,\\ 0\le t\le j}} \binom{i}{s}\binom{j}{t}\expect{Z_P^se^{-Z_P}}\expect{Z_Q^te^{-Z_Q}} \expect{Z_R^{i+j-(s+t)}e^{-2Z_R}}\nonumber\\
\ & =\sum_{\substack{0\le s\le i,\\ 0\le t\le j}} \binom{i}{s}\binom{j}{t} \frac{\sigma^{2K-r}\Gamma\left(s+\frac{K-r}{2}\right)\Gamma\left(t+\frac{K-r}{2}\right)\Gamma\left(i+j-(s+t)+\frac{r}{2}\right)}{(\sigma^2+1)^{s+t+K-r}(\sigma^2+2)^{i+j-(s+t)+\frac{r}{2}}\left(\Gamma\left(\frac{K-r}{2}\right)\right)^2\Gamma\left(\frac{r}{2}\right)},
\end{align}
where we have used the fact that if, $aU\sim \chi^2_k$ $(a>0)$, for $p,q\ge 0$, 
\begin{align}
\expect{U^pe^{-qU}} & = \int_0^\infty \left(\frac{u}{a}\right)^pe^{-uq/a}\frac{e^{-u/2}u^{k/2-1}}{\Gamma\left(\frac{k}{2}\right)2^{k/2}}du =\frac{\Gamma\left(p+\frac{k}{2}\right)}{a^p(\frac{q}{a}+\frac{1}{2})^{p+k/2}2^{k/2}\Gamma\left(\frac{k}{2}\right)}.
\end{align}
Therefore, $\sum_{\substack{S,T\subset [n]:\\ \abs{S}=\abs{T}=K}}\expect{Z_S^iZ_T^je^{-(Z_S+Z_T)}} = \sum_{r=0}^K N_r h_{ij}(r),$
where $N_r\ (0\le r\le K)$ is the number of pairs of subsets $S,T$ of $[n]$ such that $\abs{S}=\abs{T}=K$ and $\abs{S\cap T}=r$. An expression of $N_r$ can be obtained in the following way. First note that one can choose a subset $S\subset [n]$ with $\abs{S}=K$ in $\binom{n}{K}$ distinct ways. Now, for each such subset $S$, one can create a subset $T\subset [n]$ with $\abs{T}=K,\ \abs{S\cap T}=r$, by selecting $r$ elements from $S$ and $K-r$ elements from $[n]\setminus S$. Note that when $K-r\le n-K$, for each $S$, this can be done in $\binom{K}{r}\binom{n-K}{K-r}$ ways. However, if $K-r>n-K$, then one cannot find a set $T$ with $\abs{T}=K$ and $\abs{T\cap S}=r$. Therefore, the total number of ways of choosing such $S,T$, is found to be $N_r=\binom{n}{K}\binom{K}{r}\binom{n-K}{K-r}$ when $r\ge r_{\min}$ where $r_{\min}=\max\{0,2K-n\}$, and $N_r=0$ if $r<2K-n$. Consequently, using Eq.~\eqref{eq:mmse-second-term-preliminary-expression} along with the expression of $h_{ij}(r)$ from Eq.~\eqref{eq:h_ij-definition} as well as the expression of $N_r,\ 0\le r\le K$, one obtains,
\begin{align}
\expect{\left(\expect{\norm{\bvec{X}}\mid \bvec{Y}=\bvec{y}}\right)^2} 
 & =\frac{2\left(\frac{\sigma^2}{\sigma^2+1}\right)^{J+1}}{\binom{n}{K}}\sum_{\substack{r_{\min}\le r\le K,\\ i,j\ge 0}}\binom{K}{r}\binom{n-K}{K-r}a_ia_jh_{ij}(r).
\end{align}  
%
Using $\expect{\norm{\bvec{X}}^2}=n$ and plugging the above expression in Eq.~\eqref{eq:mmse-basic-expression} one obtains Eq.~\eqref{eq:mmse-expression}. 
\subsection{Proof of Theorem~\ref{thm:asymptotic-mmse}}
\label{sec:proof-thm-asymptotic-mmse}
\subsubsection{Finding $\lim_{\sigma\to 0}\texttt{mmse}$ and $\lim_{\sigma\to \infty}\texttt{mmse}$}
\label{sec:mmse-expression-with-sigma}
To evaluate $\lim_{\sigma\to 0}\texttt{mmse}$ as well as $\lim_{\sigma\to \infty}\texttt{mmse}$, we first observe from Eq.~\eqref{eq:mmse-expression-conditional-simple-expectation} that one can express the MMSE estimator, given $\bvec{B}=\bvec{b}$, as the expectation of the non-negative random variable $Z$, such that we can decompose $Z^2$ in the following way (see the description of $Z^2$ following Eq.~\eqref{eq:mmse-expression-conditional-simple-expectation}) $Z^2=\norm{\bvec{X}_2}^2 + \norm{\frac{\bvec{X}_4}{\sqrt{\sigma^2+1}} + \frac{\sigma}{\sqrt{\sigma^2+1}}\bvec{X}_3}^2$, where $\bvec{X}_2\sim \mathcal{N}(\bvec{0}_{S^C},\bvec{I}_{S^C})$ and $\bvec{X}_3,\bvec{X}_4\sim \mathcal{N}(\bvec{0}_S,\bvec{I}_S)$, where $S=\supp{\bvec{b}}$ and $\bvec{X}_2$ and $\bvec{X}_3$ are independent. This follows since $\bvec{Y}_S\sim\mathcal{N}\left(\bvec{0}_S,(\sigma^2+1)\bvec{I}_S\right)$. Therefore, the second term in the expression of the \texttt{mmse} in~\eqref{eq:mmse-basic-expression} can be expressed as $\expectsuff{S}{\expectsuff{\bvec{X}_4}{\left(\expectsuff{\bvec{X}_2,\bvec{X}_3}{Z}\right)^2}}$. Now observe that $Z^2\stackrel{a.s.}{\to} Z_1^2 = \norm{\bvec{X}_2}^2+\norm{\bvec{X}_4}^2$ as $\sigma\to 0$, and $Z^2\stackrel{a.s.}{\to}Z_2^2=\norm{\bvec{X}_2}^2+\norm{\bvec{X}_3}^2$, as $\sigma\to \infty$. It follows that $Z\stackrel{a.s.}{\to}\abs{Z_1}$ as $\sigma\to 0$ and $Z\stackrel{a.s.}{\to}\abs{Z_2}$ as $\sigma\to \infty$. Furthermore, $Z^2\le W^2=\norm{\bvec{X}_2}^2 + 2\left(\norm{\bvec{X}_4}^2 + \norm{\bvec{X}_3}^2\right)$, where $\expect{W^2}=J + 4K<\infty$. Therefore, $Z\le \abs{W}$, with $\expect{\abs{W}}<\infty$. Therefore, by the dominated convergence theorem, $\expect{Z}\to \expect{\abs{Z_1}}$ as $\sigma\to 0$ and $\expect{Z}\to \expect{\abs{Z_2}}$ as $\sigma\to \infty$.  

Now note that $Z_2^2$ is distributed as a $\chi^2_n$ random variable, so that $\expectsuff{\bvec{X}_2,\bvec{X}_3}{\abs{Z}_2}= \frac{\sqrt{2}\Gamma\left(\frac{n+1}{2}\right)}{\Gamma\left(\frac{n}{2}\right)}$, which implies from~Eq.~\eqref{eq:mmse-basic-expression} that $\lim_{\sigma\to \infty}\texttt{mmse} = n - 2\left(\frac{\Gamma\left(\frac{n+1}{2}\right)}{\Gamma\left(\frac{n}{2}\right)}\right)^2.$
On the other hand, \begin{align}
\label{eq:mmse-sigma-0-preliminary}
\lim_{\sigma\to 0}\texttt{mmse} & = n - \expectsuff{S}{\expectsuff{\bvec{X}_4}{\left(\expectsuff{\bvec{X}_2}{\abs{Z_1}\mid \bvec{X}_4}\right)^2}}.
\end{align}
In order to evaluate the RHS of Eq.~\eqref{eq:mmse-sigma-0-preliminary},
let us define the random variables $V_1=\norm{\bvec{X}_2}^2,\ V_2 = \norm{\bvec{X}_4}^2$, so that $V_1,V_2$ are distributed as $\chi^2_J$ and $\chi^2_K$ random variables respectively, and $\abs{Z_1}=\sqrt{V_1+V_2}$. Then one obtains, \begin{align}
\lefteqn{\expectsuff{\bvec{X}_4}{\left(\expectsuff{\bvec{X}_2}{\abs{Z_1}}\right)^2} = \expectsuff{V_2}{\left(\expectsuff{V_1}{\sqrt{V_1+V_2}}\right)^2} = \int_0^\infty \left(\int_0^\infty\sqrt{v_1+v_2}\frac{e^{-v_1/2}v_1^{J/2-1}}{2^{J/2}\Gamma\left(\frac{J}{2}\right)}dv_1\right)^2 \frac{e^{-v_2/2}v_2^{K/2-1}}{2^{K/2}\Gamma\left(\frac{K}{2}\right)}dv_2} & &\nonumber\\
\ &
\label{eq:mmse-sigma=0-expectation-expression} =\frac{\int_0^\infty\int_0^\infty\int_0^\infty\sqrt{(v_1+v_2)(v_3+v_2)}e^{-(v_1+v_2+v_3)/2}(v_1v_3)^{J/2-1}v_2^{K/2-1} dv_1dv_2dv_3}{2^{J+K/2}\left(\Gamma\left(\frac{J}{2}\right)\right)^2\Gamma\left(\frac{K}{2}\right)} = \frac{16I(J,K)}{\left(\Gamma\left(\frac{J}{2}\right)\right)^2\Gamma\left(\frac{K}{2}\right)},
\end{align}
where \begin{align}
I(J,K) & = \int_0^\infty \int_0^\infty \int_0^\infty \sqrt{(x^2+y^2)(z^2+y^2)}e^{-(x^2+y^2+z^2)}(xz)^{J-1}y^{K-1}dxdydz,
\end{align}
which is obtained by using the transformation of variables $v_1 = 2x^2,\ v_2=2y^2,\ v_3=2z^2$. Now, using the transformation of variables $x = r\cos a,\ y=r\sin a\cos b,\ z = r\sin a\sin b$ one can obtain, \begin{align}
\lefteqn{I(J,K)} & & \nonumber\\
\label{eq:I(J,K)-expression}
\ & = \left(\int_0^{\infty} e^{-r^2}r^{n+J+1}dr\right)\int_0^{\pi/2}\int_0^{\pi/2}\sqrt{1-\sin^2a\sin^2b}(\sin a)^n(\cos a)^{J-1} (\cos b)^{K-1} (\sin b)^{J-1} da db.
\end{align}
Now using $\sqrt{1-x}=\sum_{l\ge 0}\frac{\left(-\frac{1}{2}\right)_l}{l!}x^l$ for $\abs{x}\le 1$,\begin{align}
\lefteqn{\int_0^{\pi/2}\int_0^{\pi/2}\sqrt{1-\sin^2a\sin^2b}(\sin a)^n(\cos a)^{J-1} (\cos b)^{K-1} (\sin b)^{J-1} da db} & &\nonumber\\
\ & = \sum_{l=0}^\infty \frac{\left(-\frac{1}{2}\right)_l}{l!}\int_0^{\pi/2}\int_0^{\pi/2} (\sin a)^{n+2l}(\cos a)^{J-1} (\cos b)^{K-1}(\sin b)^{J+2l-1}da db\nonumber\\
\ & = \sum_{l=0}^\infty \frac{\left(-\frac{1}{2}\right)_l}{l!}\frac{\Gamma\left(\frac{n+2l+1}{2}\right)\Gamma\left(\frac{J}{2}\right)}{2\Gamma\left(\frac{n+2l+J+1}{2}\right)}\frac{\Gamma\left(\frac{J+2l}{2}\right)\Gamma\left(\frac{K}{2}\right)}{2\Gamma\left(\frac{n+2l}{2}\right)} = \frac{\Gamma\left(\frac{n+1}{2}\right)\left(\Gamma\left(\frac{J}{2}\right)\right)^2\Gamma\left(\frac{K}{2}\right)}{4\Gamma\left(\frac{n+J+1}{2}\right)\Gamma\left(\frac{n}{2}\right)}\sum_{l=0}^\infty \frac{\left(-\frac{1}{2}\right)_l}{l!}\frac{\left(\frac{n+1}{2}\right)_l\left(\frac{J}{2}\right)_l}{\left(\frac{n+J+1}{2}\right)_l\left(\frac{n}{2}\right)_l}\nonumber\\
\label{eq:mmse-sigma=0-angular-integral}
\ & = \frac{\Gamma\left(\frac{n+1}{2}\right)\left(\Gamma\left(\frac{J}{2}\right)\right)^2\Gamma\left(\frac{K}{2}\right)}{4\Gamma\left(\frac{n+J+1}{2}\right)\Gamma\left(\frac{n}{2}\right)}\ _3F_2\left(\frac{n+1}{2},\frac{J}{2},-\frac{1}{2};\frac{n+J+1}{2},\frac{n}{2};1\right).
\end{align}
Therefore, from Eqs.~\eqref{eq:mmse-sigma=0-expectation-expression},~\eqref{eq:I(J,K)-expression},~\eqref{eq:mmse-sigma=0-angular-integral} and using the fact that $\int_0^{\infty} e^{-r^2}r^{n+J+1}dr = \frac{1}{2}\Gamma\left(\frac{n+J}{2}+1\right)$ along with some further easy simplifications, one obtains,  \begin{align}
\lim_{\sigma\to 0}\texttt{mmse}
 & = n - (n+J)\frac{B\left(\frac{n+J}{2},\frac{n+1}{2}\right)}{B\left(\frac{n+J+1}{2},\frac{n}{2}\right)}\ _3F_2\left(\frac{n+1}{2},\frac{J}{2},-\frac{1}{2};\frac{n+J+1}{2},\frac{n}{2};1\right).
\end{align}
\subsubsection{Finding $\lim_{n\to \infty}\texttt{mmse}/n$}
\label{sec:mmse-n-ratio-large-n}
We now proceed to find an upper bound of $\texttt{mmse}$ by keeping the $K/n $ fixed to $p\ (1/n\le p\le 1)$. First observe that one can rewrite Eq.~\eqref{eq:mmse-estimator-expression} in the following form: \begin{align}
\expect{\left(\expect{\norm{\bvec{X}}\mid \bvec{Y}=\bvec{y}}\right)^2} & = 2\left(\frac{\sigma^2}{\sigma^2+1}\right)^{J+1}\expect{\left(\expect{L_S!a_{L_S}\mid \bvec{Y}=\bvec{y}}\right)^2},
\end{align} 
where $L_S$ is Poisson distributed with parameter $Z_S$, conditioned on the set $S$. Since $\texttt{mmse}=n-\expect{\left(\expect{\norm{\bvec{X}}\mid \bvec{Y}=\bvec{y}}\right)^2}$, to find an upper bound on $\texttt{mmse}$, it is enough to find a lower bound of $2\left(\frac{\sigma^2}{\sigma^2+1}\right)^{J+1}\expect{\left(\expect{L_S!a_{L_S}\mid \bvec{Y}=\bvec{y}}\right)^2}$, such that $K/n =p$. 

Now, from Eq.~\eqref{eq:al-definition}, it can be easily observed that $L_S!a_{L_S} = e^{1/(\sigma^2+1)}\expectsuff{W}{\left(\frac{J}{2}\right)_W\frac{\Gamma\left(\frac{n+1}{2}+L_S+W\right)}{\Gamma\left(\frac{n}{2}+L_S+W\right)}},$
where $W$ is a Poisson random variable (independent of $L_S$) with parameter $1/(\sigma^2+1)$. For large $n$, applying the Stirling's approximation $\Gamma(n)\approx n^ne^{-n}\sqrt{2\pi n}$, one obtains, for arbitrary non-negative integers $L,l$, $\frac{\Gamma\left(\frac{n+1}{2}+L+l\right)}{\Gamma\left(\frac{n}{2}+L+l\right)} \approx  \sqrt{\frac{\frac{n-1}{2}+L+l}{e}}\left(1+\frac{1}{2\left(\frac{n}{2}+L+l-1\right)}\right)^{\frac{n-1}{2}+L+l} \approx \left(\frac{n}{2}+L+l\right)^{1/2}.$
Therefore, for large $n$, \begin{align}
L_S!a_{L_S} & \approx e^{1/(\sigma^2+1)} \expectsuff{W}{\left(\frac{J}{2}\right)_W\sqrt{\frac{n}{2}+L_S+W}}.
\end{align}
Therefore, for large $n$, $\texttt{mmse} \approx n - 2e^{2/(\sigma^2+1)}\left(\frac{\sigma^2}{\sigma^2+1}\right)^{J+1}H_n$, where $H_n = \expect{\left(\expect{\left(\frac{J}{2}\right)_W\sqrt{\frac{n}{2}+L_S+W}\mid \bvec{Y}=\bvec{y}}\right)^2}.$ It is now enough to find a lower bound of $H_n$ with $K/n=p$. 
Before proceeding further, we note that one can apply the Taylor's expansion to get a quadratic expansion lower bound of the function $\sqrt{1+x}$ around $a$ (for arbitrary $a>0$) as $\sqrt{1+x} \ge  \sqrt{1+a}+\frac{x-a}{2\sqrt{1+a}}-\frac{(x-a)^2}{8}.$
Therefore, writing $Q = L_S+W$, 
\begin{align}
\sqrt{\frac{n}{2}+Q} & \ge \sqrt{\frac{n}{2}}\left[\sqrt{1+a}+\frac{\frac{2Q}{n}-a}{2\sqrt{1+a}}-\frac{\left(\frac{2Q}{n}-a\right)^2}{8}\right]\nonumber\\
\ \implies \expectsuff{L_S}{\sqrt{\frac{n}{2}+L_S+W}\mid Z_S} & \ge \sqrt{\frac{n}{2}}\left[\sqrt{1+a}+\frac{\expectsuff{L_S}{\frac{2L_S}{n}\mid Z_S}+\frac{2W}{n}-a}{2\sqrt{1+a}}-\frac{\expectsuff{L_S}{\left(\frac{2L_S+2W}{n}-a\right)^2\mid Z_S}}{8}\right]\nonumber\\
\ & =\sqrt{\frac{n}{2}}\left[\sqrt{1+a}+\frac{\frac{2Z_S+2W}{n}-a}{2\sqrt{1+a}}-\frac{\frac{4(Z_S^2+Z_S)}{n^2}+\frac{4Z_S}{n}\left(\frac{2W}{n}-a\right)+\left(\frac{2W}{n}-a\right)^2}{8}\right],  
\end{align}
where in the last step we have used $\expect{L_S\mid Z_S}=Z_S$ and $\expect{L_S^2\mid Z_S}=Z_S^2+Z_S$. Now, note that \begin{align}
\expectsuff{S}{Z_S\mid \bvec{Y}} & =\frac{1}{\binom{n}{K}}\sum_{\substack{S\subset [n]:\\\abs{S}=K}} \frac{\norm{\bvec{Y}_S}^2}{2\sigma^2(\sigma^2+1)} = \frac{1}{\binom{n}{K}} \frac{\binom{n-1}{K-1}\norm{\bvec{Y}}^2}{2\sigma^2(\sigma^2+1)}= \frac{K}{2n\sigma^2(\sigma^2+1)}\norm{\bvec{Y}}^2,\\
\expectsuff{S}{Z_S^2\mid \bvec{Y}} & = \frac{1}{\binom{n}{K}}\sum_{\substack{S\subset [n]:\\\abs{S}=K}} \frac{\norm{\bvec{Y}_S}^4}{4\sigma^4(\sigma^2+1)^2} = \frac{1}{4\binom{n}{K}\sigma^4(1+\sigma^2)^2}\sum_{\substack{S\subset [n]:\\\abs{S}=K}}\left(\sum_{i\in S}Y_i^4+2\sum_{\substack{i<j:\\i,j\in S}}Y_i^2Y_j^2\right)\nonumber\\
\ & = \frac{1}{4\binom{n}{K}\sigma^4(1+\sigma^2)^2} \left[\binom{n-1}{K-1}\opnorm{\bvec{Y}}{4}^4+2\binom{n-2}{K-2}\sum_{i<j}Y_i^2Y_j^2\right]\nonumber\\
\ & = \frac{1}{4\sigma^4(\sigma^2+1)^2}\left[\frac{K}{n}\opnorm{\bvec{Y}}{4}^4+\frac{2K(K-1)}{n(n-1)}\sum_{i<j}Y_i^2Y_j^2\right].
\end{align}
Furthermore, $\expect{Y_i^2} = \sigma^2+1,\ \expect{Y_i^4} = 3(\sigma^2+1)^2$,
so that 
\begin{align}
\expect{Z_S} & = \frac{K}{2\sigma^2},\ \expect{Z_S^2} = \frac{3K+K(K-1)}{4\sigma^4}= \frac{K^2+2K}{4\sigma^4}.
\end{align} 
Therefore, using Jensen's inequality, $H_n \ge  \frac{n}{2}\left(\expectsuff{W}{\left(\frac{J}{2}\right)_W\theta_W(a)}\right)^2,$
where \begin{align}
\theta_W(a) & = \sqrt{1+a}+\frac{\frac{2\expect{Z_S}+2W}{n}-a}{2\sqrt{1+a}}-\frac{\frac{4\expect{Z_S^2+Z_S}}{n^2}+\frac{4\expect{Z_S}}{n}\left(\frac{2W}{n}-a\right)+\left(\frac{2W}{n}-a\right)^2}{8}\nonumber\\
\ & = \sqrt{1+a}+\frac{\frac{p}{\sigma^2}+\frac{2W}{n}-a}{2\sqrt{1+a}}-\frac{p\left(\frac{p}{\sigma^4}+\frac{2(\sigma^2+1)}{n\sigma^4}\right)+\frac{2p}{\sigma^2}\left(\frac{2W}{n}-a\right)+\left(\frac{2W}{n}-a\right)^2}{8}.
\end{align}
Now, using the fact that $W$ is a Poisson random variable with parameter $1/(\sigma^2+1)$, one can obtain the following after some calculations,
\begin{align}
\expectsuff{W}{\left(\frac{J}{2}\right)_W\left(\frac{2W}{n}-a\right)} & = e^{-1/(\sigma^2+1)}\left(\frac{\sigma^2}{\sigma^2+1}\right)^{-J/2}\left(\frac{1-p}{\sigma^2}-a\right),\\
\expectsuff{W}{\left(\frac{J}{2}\right)_W\left(\frac{2W}{n}-a\right)^2} & = e^{-1/(\sigma^2+1)}\left(\frac{\sigma^2}{\sigma^2+1}\right)^{-J/2}\left(\left(\frac{1-p}{\sigma^2}-a\right)^2+\frac{2(1-p)(1+\sigma^2)}{n\sigma^4}\right) .
\end{align}
Therefore, $2e^{2/(\sigma^2+1)}\left(\frac{\sigma^2}{\sigma^2+1}\right)^{J+1} H_n \ge \left(\frac{n\sigma^2}{\sigma^2+1}\right)(g_n(a))^2,$
where \begin{align}
g_n(a) & = \sqrt{1+a}+\frac{\frac{p}{\sigma^2}+\frac{1-p}{\sigma^2}-a}{2\sqrt{1+a}}-\frac{\frac{p^2}{\sigma^4}+\frac{2p(1+\sigma^2)}{n\sigma^4}+\frac{2p}{\sigma^2}\left(\frac{1-p}{\sigma^2}-a\right)+\left(\frac{1-p}{\sigma^2}-a\right)^2+\frac{2(1-p)(1+\sigma^2)}{n\sigma^4}}{8}\nonumber\\
\ & = \sqrt{1+a}+\frac{\frac{1}{\sigma^2}-a}{2\sqrt{1+a}}-\frac{\left(\frac{1}{\sigma^2}-a\right)^2+\frac{2(\sigma^2+1)}{n\sigma^4}}{8}.
\end{align}
Since $a$ was chosen to be an arbitrary positive real number, we can choose $a=\frac{1}{\sigma^2}$ to obtain, $2e^{2/(\sigma^2+1)}\left(\frac{\sigma^2}{\sigma^2+1}\right)^{J+1} H_n \ge \left(\frac{n\sigma^2}{\sigma^2+1}\right)(g(1/\sigma^2))^2=n\left(1-\frac{\sqrt{\sigma^2+1}}{4n\sigma^3}\right)^2.$
Therefore, for large $n$ and large $K$ such that $K/n=p\ (1/n\le p\le 1)$, $\frac{\texttt{mmse}}{n} \le 1 - \left(1-\frac{\sqrt{\sigma^2+1}}{4n\sigma^3}\right)^2$, from which, taking $n\to \infty$, we obtain $\lim_{n\to \infty}\frac{\texttt{mmse}}{n}=0$.
\section{Conclusion}
\label{sec:conclusion}
In this paper MMSE estimation of $l_2$ norm of a random $n$ dimensional Gaussian vector is considered from AWGN perturbed measurements taken after the input vector undergoes random data missing retaining only $K\ (0\le K\le n)$ of its entries while the other entries are set to $0$. The expressions of the MMSE estimator and its MSE are derived and several asymptotic results are derived which which generalize the results in~\cite{dytso2019estimating} where $K=n$ was considered. The results show that even with large missing entries and consequently small $K$, the MMSE estimator can still be pretty close to the one obtained without the missing data phenomenon.
\appendix
\section{Derivation of Eq.~\eqref{eq:r-function-evaluation}}
\label{sec:appendix-evaluation-of-r-function}
We use the series expansion of the confluent hypergeometric function ${}_0F_1$ to obtain, 
\begin{align}
R(\alpha,\beta,\gamma,\nu,\epsilon) & = \int_0^1e^{-\alpha x}x^{\beta-1}(1-x)^{\gamma-1}\sum_{l\ge 0}\frac{(\epsilon x)^l}{(\nu)_l l!}dx \stackrel{(\zeta_1)}{=} e^{-\alpha}\int_0^1 e^{\alpha x}x^{\gamma-1}(1-x)^{\beta-1}\sum_{l\ge 0}\frac{(\epsilon (1-x))^l}{(\nu)_l l!}dx\nonumber\\
\ & \stackrel{(\zeta_2)}{=} e^{-\alpha}\int_0^1 \sum_{k,l\ge 0}\frac{\alpha^k\epsilon^l}{k!l!(\nu)_l}x^{k+\gamma-1}(1-x)^{\beta+l-1}dx \stackrel{(\zeta_3)}{=} e^{-\alpha}\sum_{k,l\ge 0}\frac{\alpha^k\epsilon^l}{k!l!(\nu)_l}\int_0^1 x^{k+\gamma-1}(1-x)^{\beta+l-1}dx\nonumber\\
\label{eq:r-function-general-expression}
\ & \stackrel{(\zeta_4)}{=} e^{-\alpha}\sum_{k,l\ge 0}\frac{\alpha^k\epsilon^l}{k!l!(\nu)_l}\frac{\Gamma(k+\gamma)\Gamma(\beta+l)}{\Gamma(\beta+\gamma+k+l)} \stackrel{(\zeta_5)}{=} e^{-\alpha}\frac{\Gamma(\gamma)\Gamma(\beta)}{\Gamma(\beta+\gamma)}\sum_{k,l\ge 0}\frac{\alpha^k\epsilon^l}{k!l!(\nu)_l}\frac{(\gamma)_k(\beta)_l}{(\beta+\gamma)_{k+l}},
\end{align}
where step $(\zeta_1)$ used the variable transformation $x\to 1-x$; step $(\zeta_2)$ used the series expansion of $e^{\alpha x}$; step $(\zeta_3)$ used Fubini's theorem to interchange the order of the integral and the summation; steps $(\zeta_4)$ and $(\zeta_5)$ used the properties of beta and gamma functions respectively. Putting $\beta=\nu$ in Eq.~\eqref{eq:r-function-general-expression} results in the expression of Eq.~\eqref{eq:r-function-evaluation}
\bibliography{estimation-gaussian-vector}
\end{document}